\documentclass[table]{ccjnl}
\usepackage{lipsum,amsmath}
\usepackage{cuted}
\usepackage{bm}
\usepackage{color}
\graphicspath{{figures/}}

\title{Capacity Analysis and Sum Rate Maximization for the SCMA Cellular Network Coexisting with D2D Communications}
\author{Yukai Liu, Wen Chen\corinfo{wenchen@sjtu.edu.cn}}

\address[]{Shanghai Institute of Advanced Communications and Data Sciences, Department of Electronic Engineering, Shanghai Jiao Tong University, Shanghai 200240, China}

\receiveddate{}
\reviseddate{}
\Editor{}

\begin{document}

\maketitle

\begin{abstract}
Sparse code multiple access (SCMA) is the most concerning scheme among non-orthogonal multiple access (NOMA) technologies for 5G wireless communication new interface. Another efficient technique in 5G aimed to improve spectral efficiency for local communications is device-to-device (D2D) communications. Therefore, we utilize the SCMA cellular network coexisting with D2D communications for the connection demand of the Internet of things (IOT), and improve the system sum rate performance of the hybrid network. We first derive the information-theoretic expression of the capacity for all users and find the capacity bound of cellular users based on the mutual interference between cellular users and D2D users. Then we consider the power optimization problem for the cellular users and D2D users jointly to maximize the system sum rate. To tackle the non-convex optimization problem, we propose a geometric programming (GP) based iterative power allocation algorithm. Simulation results demonstrate that the proposed algorithm converges fast and well improves the sum rate performance.
\end{abstract}
\section{Introduction}
\label{Introduction}
The large connection and high data rate constitute the core goals for the 5G wireless networks. Some key technologies of 5G systems include device-to-device (D2D) communications, non-orthogonal multiple access techniques (NOMA) along with massive multiple-input multiple-output (MIMO), ultra-dense radio networking, all-spectrum access, and so on \cite{1}. The conventional orthogonal multiple access (OMA) schemes allocate orthogonal resource blocks (RBs) either in time, frequency, or code domains to different users \cite{2,3}, which is underloaded as such schemes occupy more RBs than the users. However, NOMA can support massive user access via non-orthogonal RB allocation \cite{4}, which is an overloaded system. Two main approaches of the power-domain NOMA and the code-domain NOMA are widely investigated. When multiple users transmit at the same RB by power-domain NOMA, users are allocated with different power levels, while the code-domain NOMA transmits at the same RB in distinguishing spreading codes \cite{5,34}.

Among the currently proposed code domain NOMA, sparse code multiple access (SCMA) has been commonly considered as a promising method. SCMA is proposed by Nikopour and Baligh \cite{6}, in which, coded bits are directly mapped into the multi-dimensional complex lattice point (called codeword) and the codewords are designed to be sparse. The sparsity of the SCMA codewords enables massive connectivity and the use of suboptimal message passing algorithm (MPA) to detect multiple users. For general SCMA system, the users' SCMA codeword in each dimension will be modulated to an Orthogonal Frequency Division Multiple (OFDM) subcarrier before the air interface. The number of users is larger than the number of subcarriers, which lead to the non-orthogonal feature. On the other side, D2D communication is a promising technology for local communications, which allows nearby devices to communicate without base station (BS) or with limited BS involvement, and improves the link reliability, spectral efficiency and system capacity \cite{7}. It is expected that the social network structure can be further expanded through D2D communications, which supports the development of the Internet of things (IOT). Many recent works have investigated D2D communications and SCMA cellular communications hybrid network \cite{8,9,10,11,12,13,14,15,16,17,18}, and often focus on optimization problem \cite{8,10,11,12,13,14,15,16,17,18,19}.

When D2D communications are considered in uplink cellular network, two main modes have been utilized. In simple mode, a proportion of available resources are allocated to D2D users while both cellular users and D2D users utilize SCMA. In another mode, SCMA is only employed by cellular users and some subcarriers can be reused by D2D users' transmission. Liu $et$ $al.$ \cite{8} have considered the analytical model for the SCMA enhanced cellular network coexisting with D2D. Based on the approximate system area spectral efficiency (ASE), the performance between two different modes has been compared and the system design guidance along with SCMA codebook allocation scheme have been proposed. In \cite{9}, the distance between D2D transmitter and the base station has been restricted by specific methods, which improves the block error rate (BER) performance as well as the convergence behavior.

Both simple mode and complex mode have been considered into the research of optimization problem \cite{10,11,14,15}. A graph-based joint mode selection and resource allocation algorithm with predefined threshold has been proposed to maximize the system sum rate in \cite{10}, and \cite{11} has further considered more issues like partner assignment and admission control. Kim $et$ $al.$ \cite{12} have considered joint power and resource allocation by using heuristic and inner approximation method, and they have proposed a two-phase algorithm with low complexity. The Opportunistic SCMA codebook allocation (OSCA) \cite{13} and the interference-aware hyper-graph based codebook allocation (IAHCA) \cite{14,15} have been proposed to obtain the suboptimal solution and reduce the cross-tier interference between cellular users and D2D users. In OSCA, codebooks are opportunistically assigned based on the channel conditions of the cellular links and D2D links. While in IAHCA, hypergraph model is used to characterize the interference, and each RB is allowed to be shared by one cellular user and more than one D2D users. Moreover, a channel gain based mode selection criterion as well as a greedy style partner assignment scheme have been proposed in \cite{15}, while the total transmit power minimization has been considered in \cite{16}. Besides, the game based resource allocation scheme has been proposed in \cite{17}, and the block-chain-based transaction flow has been analyzed for resource allocation in \cite{18}. Authors in \cite{19} have utilized Lagrangian relaxation (LR) method and the Dinkelbach method successively for energy-efficient maximization problem.

Recent works for SCMA and NOMA capacity analysis are summarized as in \cite{5,20,21,22,31,32,33}. Le $et$ $al.$ \cite{20} have considered the low-density spreading (LDS) channels with fading and analyzed the spectral efficiency. In \cite{21}, the authors have analyzed and compared several NOMA schemes in code domain and proposed the theoretical derivations. Besides, Shental $et$ $al.$ \cite{5} have analyzed the closed form expression of low density code domain NOMA.

Moreover, Han $et$ $al.$ \cite{22} have considered the downlink SCMA systems with user grouping character. The power allocation scheme with maximum capacity has been proposed for both intragroup and intergroup case. Then, the Lagrange dual decomposition based algorithm has been proposed to find the near-optimal solution. The downlink SCMA model is also utilized in \cite{31}, where the authors have briefly analyzed the system capacity and separated the codebook assignment and power allocation. Chen $et$ $al.$ \cite{32} have derived the capacity region of uplink SCMA system based on the theory of multiple access channels and analyzed the influence between data rate, power allocation function, and resource assignment matrix. Further, the common and individual outage probability regions have been proposed and optimized with Lagrangian dual method and adaptive algorithm. In \cite{33}, the authors have compared and analyzed the channel capacities of SCMA and low-density signature multiple access (LDSMA) schemes.

In recent related work, when sum rate maximization problems are considered, concrete analysis for the expression of cellular users' capacity is lacked \cite{11,19,22}. Besides, the complex mode which brings more mutual interference between cellular users and D2D users is worth optimization for system sum rate. This paper considers cellular uplink transmission simultaneously with D2D communications, where SCMA is employed only for cellular users and D2D users work in complex mode.

In this paper, our main contribution contains the derivation of SCMA capacity bound for general SCMA system and the sum rate maximization problem. We first try to derive the information-theoretic expression of the capacity for all users in the hybrid network model. Since it is hard to derive a close form expression of cellular capacity, we find the capacity bound for cellular users based on general SCMA codeword structure. After the theoretic derivation and analysis, we focus on the sum rate maximization problem by power optimization allocation for different users. Since this is a non-convex optimization problem, we propose an algorithm based on geometric programming (GP), which transforms parts of the objective function and rewrite the problem in standard GP form to get the optimal solution with convex optimization theory. By iteration, the solution of GP problem will converge to the optimal solution of the original problem. The convergence and gain of the proposed algorithm are conducted by simulation part. The performance shows that this algorithm is appropriate and significant.

The rest of the paper is organized as follows. The network model and SCMA are described in Section~\ref{System}. The theoretic analysis of system capacity for cellular users and D2D users are presented in Section~\ref{Capacities}. In Section~\ref{Sumrate}, the optimization problem with power allocation is formulated and the GP based iterative algorithm is proposed and introduced in detail. Section~\ref{Simulation} devotes to the numerical results, and the paper finally concludes in Section~\ref{Conclusion}. 

Throughout this paper, the following notations will be used \cite{23}. The bold lowercase letter $\bm{x}$ denotes a column vector, and a matrix is represented by a bold uppercase letter $\bm{X}$. $\bm{I}$ denotes the identity matrix. The superscript $(\cdot)^{T}$ denotes matrix transpose and $(\cdot)^{*}$ denotes conjugate transpose matrix. $diag(\bm{x})$ denotes a diagonal matrix with the diagonal entries being vector $\bm{x}$. $E[x]$ denotes the expectation of a random variable $x$. For a set $A$, $|A|$ denotes the number of elements, and $A\setminus n$ denotes the set $A$ in which the element $n$ is excluded.

\section{System Model}
\label{System}
In this section, we introduce the hybrid network with cellular users and D2D users. 
The factor graph and mapping matrix are also introduced for SCMA, which is essential for capacity analysis. Then, the expression of SCMA received signal is proposed.   

\subsection{Hybrid Network}
\label{hybrid}
Figure \ref{fig1} shows a hybrid network example of a cellular coexisting with two D2D pairs \cite{23}. In our proposed network, we consider $J$ cellular users and $J_{D}$ D2D pairs. The $i$th cellular user is denoted as CU$_{i}$, and the D2D transmitter as well as the receiver of the $i$th D2D pair are denoted as DT$_{i}$ and DR$_{i}$, respectively. Cellular users are allocated with SCMA codewords spread over OFDM tones and transmitted to the air interface. While for each D2D transmitter, one of the same OFDM subcarriers is occupied.
\begin{figure}
	\centering  
	\includegraphics[width=0.9\linewidth]{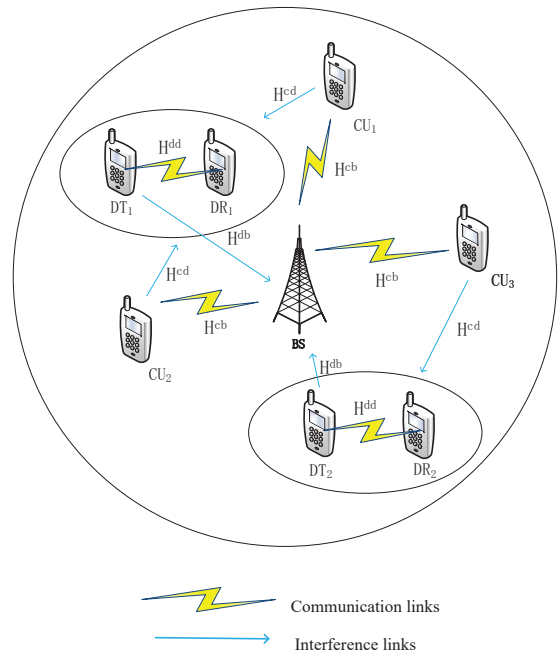}  
	\caption{The cellular and D2D hybrid network.}  
	\label{fig1}
\end{figure}

\subsection{SCMA Structure}
\label{structure}
The SCMA encoder can be defined as a mapping from $\log_{2}(M)$ coded bits to a $K$-dimensional complex codebook of size $M$. We select $N$ dimensions from $K$ for the sparse codebook design, in other words, each $K$-dimensional complex codeword has only $N$-dimensional non-zero elements. In SCMA, $J$ cellular users are allocated with $K$ OFDM subcarriers. In this model, we consider that each D2D pair will employ one of the $K$ resource blocks (OFDM subcarriers) so that $J_{D} \leq K$. For each CU$_{j}$, we introduce an indicator vector $\bm{f}_{j}=(f_{1j},f_{2j},\dots,f_{Kj})^T$ of size $K$, which $k$th elements ${f}_{kj}$ is defined as
\begin{equation}
f_{kj}=\left\{
\begin{aligned}
0,&\quad x_{jk}=0,\\
1,&\quad x_{jk}\neq 0,
\end{aligned}
\right.
\label{eq1}
\end{equation}
for $j=1,2,...,J$, where ${x}_{jk}$ is the data of user $j$ transmitting on the subcarrier $k$. $\bm{F}=(\bm{f}_{1},\bm{f}_{2},...,\bm{f}_{J})$ denotes the corresponding indicator matrix, where the set of non-zero entries in a column corresponds to the subcarriers that a cellular user will transmit over, and the set of nonzero entries in a row denotes the users who collide on the same subcarrier, 
%
The Forney factor graph can depict the structure of a indicator matrix \cite{24}. Figure \ref{fig2} shows a factor graph with $J=6$, $K=4$, and $N=2$, and Eq.~\eqref{eq2} is the corresponding indicator matrix. The size of the matrix determines the number of layer nodes and resource nodes, and when ${f}_{kj}=1$, the corresponding nodes are connected in the factor graph. Define two sets $\xi_{k}=\{j|{f}_{kj}\neq 0\}$ for $k=1,...,K$ and $\zeta_{j}=\{k|{f}_{kj}\neq 0\}$ for $j=1,...,J$. Besides, $|\xi_{k}|\triangleq d_{c}$ and $|\zeta_{j}|\triangleq d_{f}$ denote that each subcarrier supports $d_c$ users and each user occupies $d_{f}$ subcarriers, respectively. In the proposed network model, $J_{D}$ of $K$ resource nodes are also used by D2D communications. 
\begin{figure}
	\centering  
	\includegraphics[width=0.8\linewidth]{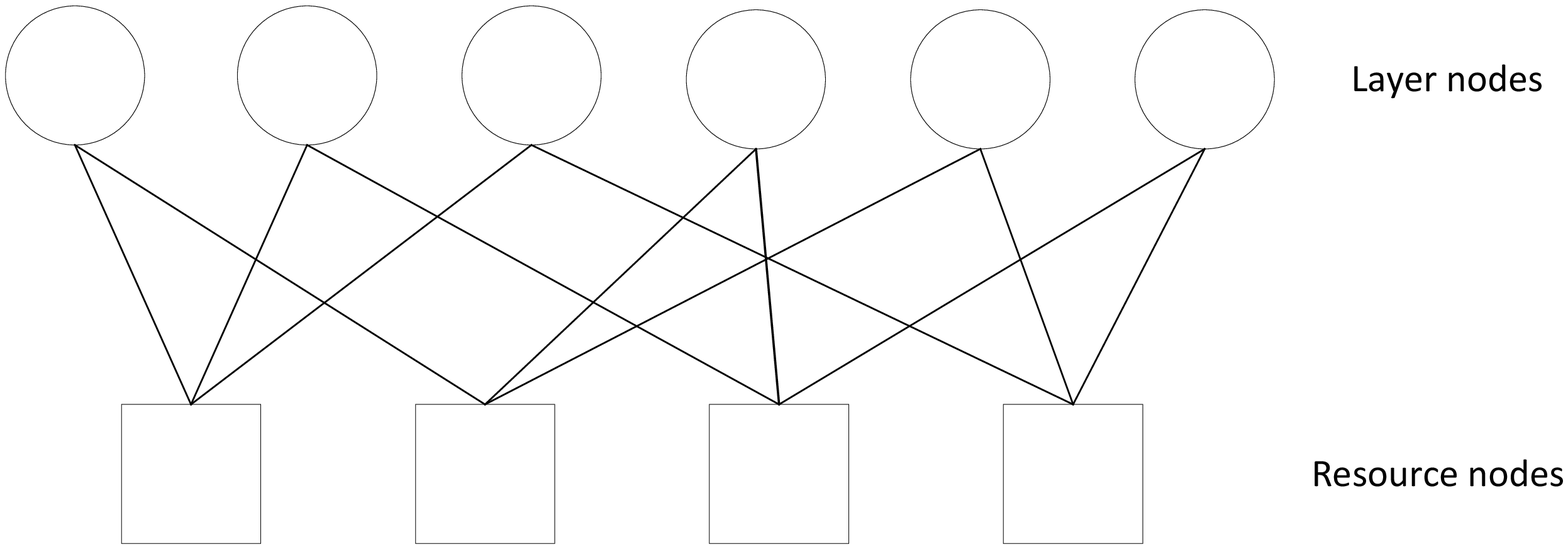}  
	\caption{Factor graph with $J=6$, $K=4$, and $N=2$.}  
	\label{fig2}
\end{figure}
\begin{equation}
\bm{F}=\left[ \begin{array}{cccccc}
1 & 1 & 1 & 0 & 0 & 0\\
1 & 0 & 0 & 1 & 1 & 0\\
0 & 1 & 0 & 1 & 0 & 1\\
0 & 0 & 1 & 0 & 1 & 1
\end{array}\right].
\label{eq2}
\end{equation}

The design of SCMA codewords has been regarded as the optimization problem for the mapping matrix and the SCMA encoder \cite{6,25}.
Let a non-zero SCMA codeword $\bm{x}_{N,j}$ be $N$-dimensional for CU$_{j}$, and $\bm{u}_{2N,j}$ be the equivalent $2N$-dimensional QAM constellation \cite{23}. The codeword can be written as
\begin{equation}
\bm{x}_{N,j}=\Delta_{j}\cdot(\bm{E}_{r}+i\cdot\bm{E}_{i})\cdot\bm{M}\cdot\bm{u}_{2N,j},
\label{eq3}
\end{equation}
where $\Delta_{j}$ denotes the constellation operators for CU$_{j}$, i.e. phase rotation, $\bm{M}$ denotes the $2N \times 2N$ unitary rotation matrix, $\bm{E}_{r}$ and $\bm{E}_{i}$ are the $N \times 2N$ matrices which can select components from vector $\bm{u}_{j}$ that corresponds to the real part and imaginary part of QAM symbols, respectively, and $i=\sqrt{-1}$ is the imaginary unit.

\subsection{SCMA Transmission}
\label{transmission}
In the SCMA enhanced cellular uplink transmission, the received signal $\bm{y}$ at BS contains cellular users' data as well as D2D users' data, expressed as
\begin{equation}
\bm{y}=\sum_{j=1}^{J}diag(\bm{h}_{j}^{cb})\bm{x}_{j}+\sum_{j=1}^{J_{D}}h^{db}_{j}\bm{x}^{'}_{j}+\bm{n},
\label{eq4}
\end{equation}
where $\bm{h}_{j}^{cb}=(h_{j1}^{cb},h_{j2}^{cb},...,h_{jK}^{cb})^{T}$ is the channel vector between CU$_{j}$ and BS, and $\bm{x}_{j}=(x_{j1},x_{j2},...,x_{jK})^{T}$ is the SCMA codeword vector of CU$_{j}$. Besides, $h^{db}_{j}$ is the channel state information (CSI) from DT$_{j}$ to BS, and $\bm{x}^{'}_{j}=(x_{j1}^{'},x_{j2}^{'},...,x_{jK}^{'})^{T}$ is the D2D codeword vector which is only non-zero in a certain row. $\bm{n}$ is the additive white Gaussian noise vector with mean $\bm{0}$ and covariance matrix $N_{0}\bm{I}$, denoted as $\bm{n}\sim\bm{\textit{CN}}(\bm{0},N_{0}\bm{I})$.
\section{Capacities of Cellular Users and D2D Users}
\label{Capacities}
In this section, we focus on the achievable rates for cellular users and D2D users and derive the information-theoretic expression, respectively. Besides, we provide specific analysis for the capacity bound of cellular users and for D2D users the closed form expression is given.

\subsection{Capacity of Cellular Users} 
\label{capacity c}
Since $\bm{n}$ is a complex Gaussian random vector, we can get the distribution for the noise in subcarrier $k$, that is, $n_{k}\sim\bm{\textit{CN}}(0,N_{0})$. Assume that subcarrier $k$ is occupied by DT$_{k}$, and define $E[x_{k}^{'}(x_{k}^{'})^{*}]\triangleq P_{k}^{'}$. Then according to Eq.~\eqref{eq4}, we have
\begin{equation}
x_{k}^{'}\sim\bm{\textit{CN}}(0,P_{k}^{'}).
\label{eq5}
\end{equation}
For cellular users' communications, the equivalent noise term includes D2D interference and noise. Thus, the equivalent noise $\tilde n_{k}\triangleq h_{k}^{db}x_{k}^{'}+n_{k}$ and
\begin{equation}
\tilde n_{k}\sim\bm{\textit{CN}}(0,N_{0}+|h_{k}^{db}|^{2}P_{k}^{'}).
\label{eq6}
\end{equation}

Let $\bm{K}_{\tilde{\bm{n}}}$ be the covariance matrix of $\tilde{\bm{n}}=(\tilde n_1,\dots,\tilde n_K)^T$, and define $\tilde N_{k}\triangleq N_{0}+|h_{k}^{db}|^{2}P_{k}^{'}$ for simplicity. Then we have $\bm{K}_{\tilde{\bm{n}}}=diag\{\tilde N_{1}, \tilde N_{2}, ..., \tilde N_{K}\}$. Rewrite Eq.~\eqref{eq4} as 
\begin{equation}
\bm{y}=\bm{H}\bm{x}+\tilde{\bm{n}},
\label{eq7}
\end{equation}
where $\bm{H}=(\bm{H}_{1},...,\bm{H}_{J})$, $\bm{H}_{j}=diag\{h_{j1}^{cb},...,h_{jK}^{cb}\}$, for $j=1,...,J$, and $\bm{x}=(\bm{x}_{1}^T,\bm{x}_{2}^T,...,\bm{x}_{J}^T)^{T}$.
%
Let $\bm{K}_{\bm{x}}$ be the covariance matrix of $\bm{x}$. Then the covariance $\bm{K}_{\bm{y}}$ of $\bm{y}$ is
\begin{equation}
\bm{K}_{\bm{y}}=\bm{K}_{\tilde{\bm{n}}}+\bm{H}\bm{K}_{\bm{x}}\bm{H}^{*}.
\label{eq8}
\end{equation}

By Shannon's theorem, the capacity is calculated based on the mutual information. The channel in Eq.~\eqref{eq7} is considered as random and time-invariant. By \cite{26}, for a specific realization, we have the following analysis process as shown in Eq.~\eqref{eq9},
\begin{equation}
\begin{aligned}
&I(\bm{x};\bm{y}|\bm{H}=\textup{H})\\
&=h(\bm{y}|\bm{H}=\textup{H})-h(\tilde{\bm{n}})\\
&\leq \log\det(\pi e\bm{K}_{\bm{y}})-\log\{(\pi e)^{K}\prod_{k=1}^{K}\tilde{N}_{k}\}\\
&=\log\det(\bm{K}_{\bm{y}})-\sum_{k=1}^{K}\log(N_{0}+|h_{k}^{db}|^{2}P_{k}^{'}),
\end{aligned}
\label{eq9}
\end{equation}
where the equality in Eq.~\eqref{eq9} holds when the input $\bm{x}$ follows Gaussian input. The maximum mutual information is not a closed form expression as the determinant is not explicit. Let $\bm{K}_{\bm{x}_{j}}=E[\bm{x}_{j}\bm{x}_{j}^{*}]$. Then it is easy to show that $\bm{K}_{\bm{x}_{j}}$ is a Hermitian matrix and
$\bm{K}_{\bm{x}}=diag\{\bm{K}_{\bm{x}_{1}},\cdots,\bm{K}_{\bm{x}_{J}}\}$
is a block diagonal matrix. 

Notice that $\bm{K}_{\bm{x}_{j}}$ is not a diagonal matrix in general, it is valuable to derive the capacity bounds with Eq.~\eqref{eq9}. We introduce two lemmas for further analysis.

\begin{lemma}
	\label{lemma1}
	Let the $N$-dimensional square matrices $\bm{Q_{1}}$, $\bm{Q_{2}}$ be Hermitian. The eigenvalues of $\bm{Q_{1}}$, $\bm{Q_{2}}$, and $\bm{Q_{1}}+\bm{Q_{2}}$ are defined as ${\lambda_{n}(\bm{Q_{1}})}_{n=1}^{N}$, ${\lambda_{n}(\bm{Q_{2}})}_{n=1}^{N}$, and ${\lambda_{n}(\bm{Q_{1}}+\bm{Q_{2}})}_{n=1}^{N}$, which are ordered as $\lambda_{1}\leq \lambda_{2}\leq...\leq\lambda_{N}$. Then for $n=1,...,N$
	\begin{equation}
	\begin{aligned}
	\lambda_{n}(\bm{Q_{1}})+\lambda_{1}(\bm{Q_{2}})&\leq\lambda_{n}(\bm{Q_{1}}+\bm{Q_{2}})\\
	&\leq\lambda_{n}(\bm{Q_{1}})+\lambda_{N}(\bm{Q_{2}}).
	\end{aligned}
	\label{eq10}
	\end{equation}
\end{lemma}
\begin{lemma}
	\label{lemma2}
	Let the $N$-dimensional square matrices $\bm{Q}$, $\bm{S}$ be
	Hermitian and nonsingular, respectively. The eigenvalues of $\bm{Q}$ and $\bm{S}\bm{Q}\bm{S}^{*}$ are defined as ${\lambda_{n}(\bm{Q})}_{n=1}^{N}$ and ${\lambda_{n}(\bm{S}\bm{Q}\bm{S}^{*})}_{n=1}^{N}$. The singular values of $\bm{S}$ are defined as $0<\sigma_{1}\leq\sigma_{2}\leq...\leq\sigma_{N}$. For each $n=1,...,N$, there is a real number $\theta_{n}$ $\in$ $[\sigma_{1}^{2},\sigma_{N}^{2}]$ such that
	\begin{equation}
	\lambda_{n}(\bm{S}\bm{Q}\bm{S}^{*})=\theta_{n}\lambda_{n}(\bm{Q}).
	\label{eq11}
	\end{equation}
\end{lemma}
Notice that Lemma~\ref{lemma1} is the classic Weyl's Theorem and Lemma~\ref{lemma2} can be easily derived by Lemma~\ref{lemma1} in \cite{27}, the proof is omitted here.

In order to find the upper bound and the lower bound of $\log\det(\bm{K}_{\tilde{\bm{n}}}+\bm{H}\bm{K}_{\bm{x}}\bm{H}^{*})$ by Eq.~\eqref{eq8}, \eqref{eq9}, we need to analyze the eigenvalues. 
Define ${\lambda_{k}(\bm{Q})}_{k=1}^{K}$ as the eigenvalues of an $K$-dimensional square matrix $\bm{Q}$ in nondecreasing order, then we produce the analysis in three aspects.
\subsubsection{Upper Bound}
\label{upper}
By using Lemma~\ref{lemma1}, we have
\begin{equation}
\begin{aligned}
\lambda_{k}(\bm{K}_{\tilde{\bm{n}}}&+\bm{H}\bm{K}_{\bm{x}}\bm{H}^{*})\\
&\leq\lambda_{k}(\bm{K}_{\tilde{\bm{n}}})+\lambda_{K}(\bm{H}\bm{K}_{\bm{x}}\bm{H}^{*})\\
&=\tilde{N}_{k}+\lambda_{K}\left(\sum_{j=1}^{J}\bm{H}_{j}\bm{K}_{\bm{x}_{j}}\bm{H}_{j}^{*}\right)\\
&\leq \tilde{N}_{k}+\sum_{j=1}^{J}\lambda_{K}(\bm{H}_{j}\bm{K}_{\bm{x}_{j}}\bm{H}_{j}^{*}).\\
\end{aligned}
\label{eq12}
\end{equation}
Notice that the first equality in Eq.~\eqref{eq12} occurs if and only if there is nonzero vector $\bm{\alpha}$ such that $\bm{K}_{\tilde{\bm{n}}}\bm{\alpha}=\lambda_{k}(\bm{K}_{\tilde{\bm{n}}})\bm{\alpha}$, $\bm{H}\bm{K}_{\bm{x}}\bm{H}^{*}\bm{\alpha}=\lambda_{K}(\bm{H}\bm{K}_{\bm{x}}\bm{H}^{*})\bm{\alpha}$, and $(\bm{K}_{\tilde{\bm{n}}}+\bm{H}\bm{K}_{\bm{x}}\bm{H}^{*})\bm{\alpha}=\lambda_{k}(\bm{K}_{\tilde{\bm{n}}}+\bm{H}\bm{K}_{\bm{x}}\bm{H}^{*})\bm{\alpha}$. In general condition, $\bm{K}_{\tilde{\bm{n}}}$ and $\bm{H}\bm{K}_{\bm{x}}\bm{H}^{*}$ should have no common eigenvector, then the inequalities in Eq.~\eqref{eq12} are strict inequalities.

Since $\bm{K}_{\bm{x}_{j}}$ is Hermitian, $\bm{H}_{j}$ is nonsingular, and $\bm{H}_{j}\bm{H}_{j}^{*}=diag\{|h_{j1}^{cb}|^{2},|h_{j2}^{cb}|^{2}...,|h_{jK}^{cb}|^{2}\}$, by using Lemma~\ref{lemma2}, for CU$_{j}$ we have  
\begin{equation}
\lambda_{K}(\bm{H}_{j}\bm{K}_{\bm{x}_{j}}\bm{H}_{j}^{*})\leq\lambda_{K}(\bm{K}_{\bm{x}_{j}})\cdot\max\limits_{k}(|h_{jk}^{cb}|^{2}).
\label{eq13}
\end{equation}
According to the SCMA codebook design in Eq.~\eqref{eq3}, we introduce a $K \times N$ matrix $\bm{V}_{j}$ to transform non-zero $N$-dimensional codeword into $K$-dimensional codeword with $K-N$ zero elements. Then we have
\begin{equation}
\begin{aligned}
\bm{x}_{j}&=\bm{V}_{j}\cdot e^{i\theta_{j}}\cdot(\bm{E}_{r}+i\cdot\bm{E}_{i})\cdot\bm{M}\cdot\bm{u}_{j}\\
&=e^{i\theta_{j}}\cdot\bm{V}_{j}\cdot\bm{M^{'}}\cdot\bm{u}_{j},
\end{aligned}
\label{eq14}
\end{equation}
where $e^{i\theta_{j}}$ denotes the phase rotation here and $\bm{M}^{'}\triangleq(\bm{E}_{r}+i\cdot\bm{E}_{i})\cdot\bm{M}$ denotes a $N \times 2N$ complex matrix. Then
\begin{equation}
\begin{aligned}
\bm{K}_{\bm{x}_{j}}&=E[e^{i\theta_{j}}(\bm{V}_{j}\bm{M^{'}}\bm{u}_{j})(\bm{V}_{j}\bm{M^{'}}\bm{u}_{j})^{*}e^{-i\theta_{j}}]\\
&=(\bm{V}_{j}\bm{M^{'}})\cdot E[\bm{u}_{j}(\bm{u}_{j})^{T}]\cdot(\bm{V}_{j}\bm{M^{'}})^{*}.
\end{aligned}
\label{eq15}
\end{equation}
Let $\bm{V}_{j}\bm{M^{'}}\triangleq(\bm{M}_{j1},\bm{M}_{j2})$ where $\bm{M}_{j1}$ and $\bm{M}_{j2}$ are $K\times N$ matrices. Besides, let $E[\bm{u}_{j}(\bm{u}_{j})^{T}]\triangleq diag[\bm{A}_{j1},\bm{A}_{j2}]$ where $\bm{A}_{j1}$ and $\bm{A}_{j2}$ are $N\times N$ diagonal matrices in Gaussian input. Then we have
\begin{equation}
\bm{K}_{\bm{x}_{j}}=\bm{M}_{j1}\bm{A}_{j1}\bm{M}_{j1}^{*}+\bm{M}_{j2}\bm{A}_{j2}\bm{M}_{j2}^{*}.
\label{eq16}
\end{equation}
By using Lemma~\ref{lemma1} repeatedly, we have
\begin{equation}
\begin{aligned}
&\lambda_{K}(\bm{K}_{\bm{x}_{j}})\\
&=\lambda_{K}(\bm{M}_{j1}\bm{A}_{j1}\bm{M}_{j1}^{*}+\bm{M}_{j2}\bm{A}_{j2}\bm{M}_{j2}^{*})\\
&\leq\lambda_{K}(\bm{M}_{j1}\bm{A}_{j1}\bm{M}_{j1}^{*})+\lambda_{K}(\bm{M}_{j2}\bm{A}_{j2}\bm{M}_{j2}^{*})\\
&\triangleq\tau_{K,j}.
\end{aligned}
\label{eq17}
\end{equation}
Combing Eq.~\eqref{eq12}, \eqref{eq13}, and \eqref{eq17}, we conclude that
\begin{equation}
\begin{aligned}
\lambda_{k}(\bm{K}_{\tilde{\bm{n}}}&+\bm{H}\bm{K}_{\bm{x}}\bm{H}^{*})\\
&<\tilde{N}_{k}+\sum_{j=1}^{J}\tau_{K,j}\cdot\max\limits_{k}(|h_{jk}^{cb}|^{2}).
\end{aligned}
\label{eq18} 
\end{equation}
This formula reveals the theoretic upper bound of each eigenvalue.

\subsubsection{Lower Bound}
\label{lower}
Similarly, we can analyze the lower bound for $\lambda_{k}(\bm{K}_{\tilde{\bm{n}}}+\bm{H}\bm{K}_{\bm{x}}\bm{H}^{*})$. By using Lemma~\ref{lemma1} and Lemma~\ref{lemma2}, we have
\begin{equation}
\begin{aligned}
\lambda_{k}(\bm{K}_{\tilde{\bm{n}}}&+\bm{H}\bm{K}_{\bm{x}}\bm{H}^{*})\\
&\geq \lambda_{k}(\bm{K}_{\tilde{\bm{n}}})+\lambda_{1}(\bm{H}\bm{K}_{\bm{x}}\bm{H}^{*})\\
&\geq \tilde{N}_{k}+\sum_{j=1}^{J}\lambda_{1}(\bm{H}_{j}\bm{K}_{\bm{x}_{j}}\bm{H}_{j}^{*})\\
&\geq \tilde{N}_{k}+\sum_{j=1}^{J}[\lambda_{1}(\bm{K}_{\bm{x}_{j}})\cdot\min\limits_{k}(|h_{jk}^{cb}|^{2})].\\
\end{aligned}
\label{eq19}
\end{equation}
By Lemma~\ref{lemma1} and Eq.~\eqref{eq16}, we have
\begin{equation}
\begin{aligned}
&\lambda_{1}(\bm{K}_{\bm{x}_{j}})\\
&=\lambda_{1}(\bm{M}_{j1}\bm{A}_{j1}\bm{M}_{j1}^{*}+\bm{M}_{j2}\bm{A}_{j2}\bm{M}_{j2}^{*})\\
&\geq\lambda_{1}(\bm{M}_{j1}\bm{A}_{j1}\bm{M}_{j1}^{*})+\lambda_{1}(\bm{M}_{j2}\bm{A}_{j2}\bm{M}_{j2}^{*})\\
&\triangleq\tau_{1,j}.
\end{aligned}
\label{eq20}
\end{equation}
Thus, we conclude in general condition that
\begin{equation}
\begin{aligned}
\lambda_{k}(\bm{K}_{\tilde{\bm{n}}}&+\bm{H}\bm{K}_{\bm{x}}\bm{H}^{*})\\
&>\tilde{N}_{k}+\sum_{j=1}^{J}\tau_{1,j}\cdot\min\limits_{k}(|h_{jk}^{cb}|^{2}).
\end{aligned} 
\label{eq21}
\end{equation}

\subsubsection{Capacity Bound}
\label{bound}
We have
\begin{equation}
\begin{aligned}
C_{c}&=\max\limits_{\bm{K}_{\bm{x}}} I(\bm{x};\bm{y}|\bm{H}=\textup{H})\\
&=\log\det(\bm{K}_{\bm{y}})-\sum_{k=1}^{K}\log(\tilde{N}_{k})\\
&=\sum_{k=1}^{K}\log\lambda_{k}(\bm{K}_{\bm{y}})-\sum_{k=1}^{K}\log(\tilde{N}_{k}),
\end{aligned}
\label{eq22}
\end{equation}
where the derivation is based on $\bm{x}\sim\textit{CN}(0,\bm{K}_{\bm{x}})$.

According to Eq.~\eqref{eq14}, the arbitrariness of $\bm{V}_{j}$ and $\bm{M}^{'}$ makes the $\bm{K}_{\bm{x}_{j}}$ non-diagonal, however, with Gaussian input we can define $\bm{u}_{j}\sim\textit{CN}(0,\bm{K}_{\bm{u}_{j}})$ where $\bm{K}_{\bm{u}_{j}}=diag\{P_{j,1}^{u},P_{j,2}^{u},...,P_{j,2N}^{u}\}$. 
Then according to Eq.~\eqref{eq17} we further analyze that
\begin{equation}
\begin{aligned}
&\tau_{K,j}\\
&=\lambda_{K}(\bm{M}_{j1}\bm{A}_{j1}\bm{M}_{j1}^{*})+\lambda_{K}(\bm{M}_{j2}\bm{A}_{j2}\bm{M}_{j2}^{*})\\
&=\lambda_{K}(\bm{M}_{j1}diag\{P_{j,1}^{u},...,P_{j,N}^{u}\}\bm{M}_{j1}^{*})\\
&+\lambda_{K}(\bm{M}_{j2}diag\{P_{j,N+1}^{u},...,P_{j,2N}^{u}\}\bm{M}_{j2}^{*}).
\end{aligned}
\label{eq23}
\end{equation}

To maximize the rate of cellular network, the upper bound of SCMA capacity should be figured out based on Eq.~\eqref{eq18}, \eqref{eq22}, and \eqref{eq23}, which is summarized in follows
\begin{equation}
\begin{aligned}
C_{c}<\sum_{k=1}^{K}\log \left\{1+\frac{\sum_{j=1}^{J}\tau_{K,j}\cdot\max\limits_{k}(|h_{jk}^{cb}|^{2})}{N_{0}+|h_{k}^{db}|^{2}P_{k}^{'}}\right\},
\label{eq24}
\end{aligned}
\end{equation}
where we should notice that, the bound for proposed network with general SCMA system is also determined by SCMA codebook design and factor graph structure. For more practical situation, SCMA codebook is designed based on optimal rotation matrix \cite{28}. Thus, the analysis of optimization problem formulation in section IV needs some simplification. 

\subsection{Capacity of D2D users}
\label{capacity d}
For D2D users, each D2D pair is allocated with one of the $K$ OFDM tones. For example, we assume that the $\ell$th D2D pair selects the $\ell$th OFDM tone for $\ell=1,2,...,J_{D}$. The received signal of DR$_\ell$ is disturbed by some cellular users' data, which is shown as
\begin{equation}
y_{\ell}=h_{\ell}^{dd} x_{\ell}^{'}+\sum_{j\in\xi_{\ell}}h_{j\ell}^{cd} x_{j\ell}+n_{\ell},
\label{eq25}
\end{equation}
where $h_{\ell}^{dd}$ is the CSI between $\ell$th D2D pair, $h_{j\ell}^{cd}$ is the CSI between CU$_{j}$ and DR$_{\ell}$, and the noise is still defined as $n_{\ell}\sim\bm{\textit{CN}}(0,N_{0})$. 
Assume that
\begin{equation}
x_{j\ell}\sim\bm{\textit{CN}}(0,P_{j\ell}),
\label{eq26}
\end{equation}
where $P_{j\ell}\triangleq E[x_{j\ell}x_{j\ell}^{*}]$. Let the equivalent noise term $n_{\ell}^{''}\triangleq\sum_{j\in\xi_{\ell}}h_{j\ell}^{cd} x_{j\ell}+n_{\ell}$ with the interference term. Then we have
\begin{equation}
n_{\ell}^{''}\sim\bm{\textit{CN}}(0,N_{0}+\sum_{j\in\xi_{\ell}}|h_{j\ell}^{cd}|^{2}P_{j\ell}).
\label{eq27}
\end{equation}
We can rewrite the received signal at DR$_{\ell}$ as
\begin{equation}
y_{\ell}=h_{\ell}^{dd}x_{\ell}^{'}+n_{\ell}^{''}.
\label{eq28}
\end{equation}
Similarly, we can calculate the mutual information for D2D pair $\ell$ as
\begin{equation}
\begin{aligned}
&I(x_{\ell}^{'};y_{\ell}|h_{\ell}^{dd})=h(y_{\ell}|h_{\ell}^{dd})-h(n_{\ell}^{''})\\
&\leq \log\left\{\pi e\left(|h_{\ell}^{dd}|^{2}P_{\ell}^{'}+N_{0}+\sum_{j\in\xi_{\ell}}|h_{j\ell}^{cd}|^{2}P_{j\ell}\right)\right\}\\
&-\log\left\{\pi e\left(N_{0}+\sum_{j\in\xi_{\ell}}|h_{j\ell}^{cd}|^{2}P_{j\ell}\right)\right\}\\
&= \log\left\{1+\frac{|h_{\ell}^{dd}|^{2}P_{\ell}^{'}}{N_{0}+\sum_{j\in\xi_{\ell}}|h_{j\ell}^{cd}|^{2}P_{j\ell}}\right\}.
\end{aligned}
\label{eq29}
\end{equation}
Thus, when the $\ell$th OFDM tone is chosen for the $\ell$th D2D pair, according to Eq.~\eqref{eq29}, we have
\begin{equation}
\begin{aligned}
C_{d}&=\sum_{\ell=1}^{J_{D}}\log\left\{1+\frac{|h_{\ell}^{dd}|^{2}P_{\ell}^{'}}{N_{0}+\sum_{j\in\xi_{\ell}}|h_{j\ell}^{cd}|^{2}P_{j\ell}}\right\}\\
&=\sum_{\ell=1}^{J_{D}}\log\left(1+\gamma_{\ell}^{d}\right),
\end{aligned}
\label{eq30}
\end{equation}
where $\gamma_{\ell}^{d}$ denotes the signal-to-interference-plus-noise ratio (SINR) for the $\ell$th D2D receiver.
\section{Sum Rate Maximization}
\label{Sumrate}

In this section, we focus on the sum rate maximization issue. We formulate the optimization problem as the power allocation of cellular users and D2D users. The GP based iterative algorithm is introduced to find the numerical solution.
\subsection{Problem Formulation}
\label{problem}
In section~\ref{Capacities}, the analysis is based on general SCMA codebook. However, according to the sparsity of SCMA codeword, most non-diagonal entries of $\bm{K}_{\bm{x}_{j}}$ are zero. Actually, most SCMA system models utilize optimal codebook design principles to obtain the diagonal covariance matrix \cite{24,25,36}. In this section, we consider the general situation as
\begin{equation}
\begin{aligned}
\bm{K}_{\bm{x}_{j}}&=E[\bm{x}_{j}\bm{x}_{j}^{*}]\\
&=diag\{E[x_{j1}x_{j1}^{*}], \dots, E[x_{jK}x_{jK}^{*}]\}\\
&\triangleq diag\{P_{j1}, P_{j2},...,P_{jK}\}.
\end{aligned}
\label{eq31}
\end{equation}
Notice that the maximum mutual information in Eq.~\eqref{eq9} can get closed form expression, we have the simplified derivation as
\begin{equation}
\lambda_{k}(\bm{K}_{y})=\tilde{N}_{k}+\sum_{j\in\xi_{k}}|h_{jk}^{cb}|^{2}P_{jk},k=1,2,...,K,
\label{eq32}
\end{equation}
where we should notice that $P_{jk}=0$ for $j\notin\xi_{k}$ with the corresponding factor graph. Then we have
\begin{equation}
\begin{aligned}
&\log\det(\bm{K}_{y})\\
&=\sum_{k=1}^{K}\log\left(N_{0}+|h_{k}^{db}|^{2}P_{k}^{'}+\sum_{j\in\xi_{k}}|h_{jk}^{cb}|^{2}P_{jk}\right).
\end{aligned}
\label{eq33}
\end{equation}
Finally, the capacity of cellular users is
\begin{equation}
C_{c}=\sum_{k=1}^{K}\log\left\{1+\frac{\sum_{j\in\xi_{k}}|h_{jk}^{cb}|^{2}P_{jk}}{N_{0}+|h_{k}^{db}|^{2}P_{k}^{'}}\right\}.
\label{eq34}
\end{equation}
Besides, for decoding the codewords of the CU$_{j}$ at the BS, we define the SINR in subcarrier $k\in\zeta_{j}$ as
\begin{equation}
\gamma_{jk}^{c}=\frac{|h_{jk}^{cb}|^{2}P_{jk}}{N_{0}+|h_{k}^{db}|^{2}P_{k}^{'}}.
\label{eq35}
\end{equation}
Accordingly, the sum rate maximization problem is formulated as
\begin{equation}
\begin{aligned}
\bm{P1:} \quad & \max_{\bm{P},\bm{P^{'}}} (C_{c}+C_{d})\\
s.t.\quad & C1:\gamma_{jk}^{c}\geq\gamma_{0}^{c},\quad j=1,\dots,J,\quad k\in\zeta_{j},\\
\quad& C2:\gamma_{\ell}^{d}\geq\gamma_{0}^{d},\quad \ell=1,\dots,J_{D},\\
\quad& C3:0\leq P_{j}=\sum_{k\in\zeta_{j}}P_{jk}\leq P_{0},\quad j=1,\dots,J,\\
\quad& C4:0\leq P_{\ell}^{'}\leq P_{0}^{'},\quad \ell=1,\dots,J_{D},
\end{aligned}
\label{eq36}
\end{equation}
where the constraints $C1$ and $C2$ denote the minimum QoS requirements of cellular users and D2D users, determined by the SINR threshold $\gamma_{0}^{c}$ and $\gamma_{0}^{d}$, respectively, $C3$ and $C4$ mean the transmitted power limitation. Since the objective function is nonlinear and nonconvex, P1 is a nonconvex optimization problem. We decide to rely on the tractable Geometric Programming problem to find the appropriate algorithm.

\subsection{Geometric Programming}
\label{GP}
Geometric Programing (GP) is a special optimization problem with useful properties \cite{29}. The standard form of GP is nonlinear and nonconvex, while the convex form of GP can make use of the standard convex optimization theory. Define a monomial $g$ as
\begin{equation}
g(\bm{x})=cx_{1}^{a_{1}}x_{2}^{a_{2}}...x_{n}^{a_{n}},
\label{eq37}
\end{equation}
where $c\geq0$ is the multiplicative constant, $a_{\ell}\in \bm{R}$ for $\ell=1,2,...,n$ is the exponents, and the argument $x_{\ell}>0$ for $\ell=1,2,...,n$. Then a posynomial is defined as a summation of monomials shown as
\begin{equation}
G(\bm{x})=\sum_{m=1}^{M}c_{m}x_{1}^{a_{1}^{m}}x_{2}^{a_{2}^{m}}...x_{n}^{a_{n}^{m}},
\label{eq38}
\end{equation}
where $a_{\ell}^{m}\in\bm{R}$ and $c_{m}\geq0$ for $\ell=1,2,\dots,n$ and $m=1,2,\dots,M$. The standard form for a GP is as follows.
\begin{equation}
\begin{aligned}
& \min_{\bm{x}} \quad G_{0}(\bm{x})\\
s.t.\quad & G_{s}(\bm{x})\leq 1, \quad s=1,\dots,S,\\
& h_{t}(\bm{x})=1, \quad t=1,\dots,T,
\end{aligned}
\label{eq39}
\end{equation}
where the posynomial $G_{s}(\cdot)$ and the monomial $h_{t}(\cdot)$ for $s=0,1,\dots, S$ and $t=1,\dots,T$ are defined as
\begin{equation}
G_{s}(\bm{x})=\sum_{m=1}^{M_{s}}c_{sm}x_{1}^{a_{1}^{sm}}x_{2}^{a_{2}^{sm}}...x_{n}^{a_{n}^{sm}},
\label{eq40}
\end{equation}
and
\begin{equation}
h_{t}(\bm{x})=c_{t}x_{1}^{a_{1}^{t}}x_{2}^{a_{2}^{t}}...x_{n}^{a_{n}^{t}}
\label{eq41}.
\end{equation}
In Eq.~\eqref{eq40} and Eq.~\eqref{eq41}, we have $c_{sm}>0$ and $c_{t}>0$ for $m=1,2,...,M_{s}$, $s=0,1,...,S$, and $t=1,2,...,T$.

For the standard form of GP, the posynomials are not convex functions. In order to simplify the problem, the logarithmic conversion as $y_{\ell}=\log x_{\ell}$, $b_{sm}=\log c_{sm}$, and $b_{t}=\log c_{t}$ are utilized. Then the equivalent problem in term of $\bm{y}$ is
\begin{equation}
\begin{aligned}
& \min_{\bm{y}} \quad G_{0}^{'}(\bm{y})=\log\sum_{m=1}^{M_{0}}\exp(\bm{a}_{0m}^{T}\bm{y}+b_{0m})\\
&s.t.\quad G_{i}^{'}(\bm{y})=\log\sum_{m=1}^{M_{s}}\exp(\bm{a}_{sm}^{T}\bm{y}+b_{sm})\leq 0,\\
&\quad\quad s=1,2,\dots,S,\\
&\quad h_{t}^{'}(\bm{y})=\bm{a}_{t}^{T}\bm{y}+b_{t}=0, \quad t=1,2,\dots,T,
\end{aligned}
\label{eq42}
\end{equation}
where $\bm{a}_{sm}=(a_{1}^{sm},\cdots,a_{n}^{sm})^{T}$ and $\bm{a}_{t}=(a_{1}^{t},\cdots,a_{n}^{t})^{T}$. The problem in \eqref{eq42} is the convex form of GP.

\subsection{The Proposed GP Based Iterative Algorithm}
The objective function in $\bm{P1}$ is obviously neither a posynomial nor a convex function. We consider the deformation of the objective function to construct the GP problem. For the transmitted power of cellular users, we further improve the constraint $C3$ as $P_{jk}\leq\frac{P_{0}}{d_f}, \forall k\in\zeta_{j}$, $\forall j$. According to Eq.~\eqref{eq30}, \eqref{eq34}, we have

\begin{equation}
\begin{aligned}
&C_{c}+C_{d}\\
&=\sum_{k=1}^{K}\log(1+\sum_{j\in\xi_{k}}{\gamma_{jk}^{c}})+\sum_{\ell=1}^{J_{D}}\log\left\{1+\gamma_{\ell}^{d}\right\}\\
&=\log\left\{\prod_{k=1}^{K}\frac{g_{k}(\bm{p},\bm{p^{'}})}{f_{k}(\bm{p^{'}})}\cdot\prod_{\ell=1}^{J_{D}}\frac{g_{\ell}^{'}(\bm{p},\bm{p^{'}})}{f_{\ell}^{'}(\bm{p})}\right\},\\
\end{aligned}
\label{eq43}
\end{equation}
where 
\begin{equation}
\begin{aligned}
&f_{k}(\bm{p^{'}})\triangleq N_{0}+|h_{k}^{db}|^{2}P_{k}^{'},\\
&f_{\ell}^{'}(\bm{p})\triangleq N_{0}+\sum_{j\in\xi_{\ell}}|h_{j\ell}^{cd}|^{2}P_{j\ell},\\
&g_{k}(\bm{p},\bm{p^{'}})\triangleq N_{0}+|h_{k}^{db}|^{2}P_{k}^{'}+\sum_{j\in\xi_{k}}|h_{jk}^{cb}|^{2}P_{jk},\\
&g_{\ell}^{'}(\bm{p},\bm{p^{'}})\triangleq N_{0}+\sum_{j\in\xi_{\ell}}|h_{j\ell}^{cd}|^{2}P_{j\ell}+|h_{\ell}^{dd}|^{2}P_{\ell}^{'}.\\
\end{aligned}
\label{eq44}
\end{equation}

Note that the logarithmic function is a monotonic function, we can change the maximization of the objective function to the minimization of its reciprocal. Then problem $\bm{P1}$ can be transformed into
\begin{equation}
\begin{aligned}
\bm{P2:} \quad& \min_{\bm{p},\bm{p^{'}}}\quad \prod_{k=1}^{K}\prod_{\ell=1}^{J_{D}}\frac{f_{k}(\bm{p^{'}})f_{\ell}^{'}(\bm{p})}{g_{k}(\bm{p},\bm{p^{'}})g_{\ell}^{'}(\bm{p},\bm{p^{'}})},\\
s.t.\quad & C1:(\gamma_{jk}^{c})^{-1}\gamma_{0}^{c}\leq 1, j=1,\dots,J, k\in\zeta_{j},\\
& C2:(\gamma_{\ell}^{d})^{-1}\gamma_{0}^{d}\leq 1, \ell=1,\dots,J_{D},\\
& C3:P_{jk}(\frac{P_{0}}{d_f})^{-1}\leq 1, j=1,\dots,J,k\in\zeta_{j},\\
& C4:P_{\ell}^{'}(P_{0}^{'})^{-1}\leq 1, \ell=1,\dots,J_{D},
\end{aligned}
\label{eq45}
\end{equation}
where all the constraints are rewritten as posynomials with upper bounds.

The transformed objective function is a ratio between two posynomials, which belongs to Complementary GP and is an intractable NP-hard problem \cite{29}. However, the denominator posynomial can be approximated by a monomial such that the ratio is a posynomial as shown in Lemma~\ref{lemma3}.
\begin{lemma}
	\label{lemma3}
	Let the $g(\bm{x})=\sum_{\ell}u_{\ell}(\bm{x})$ is a posynomial where $u_{\ell}(\bm{x})$ is a monomial. The argument $\bm{x}$ is with all positive elements. There must have a monomial $\tilde{g}(\bm{x})$ as
	\begin{equation}
	g(\bm{x})\geq \tilde{g}(\bm{x})\triangleq \prod_{\ell}\left(\frac{u_{\ell}(\bm{x})}{\beta_{\ell}}\right)^{\beta_{\ell}},
	\label{eq46}
	\end{equation}
	where $\beta_{\ell}>0$ and $\sum_{\ell}\beta_{\ell}=1$, and that becomes an equality when $\beta_{\ell}=\frac{u_{\ell}(\bm{x})}{g(\bm{x})}$ for $\forall \ell$.
\end{lemma}
\begin{proof}
	\label{proof}
	See Appendix.
\end{proof}

By using Lemma~\ref{lemma3}, we can rely on GP to find the globally optimal solution. $\bar{\bm{p}}\triangleq(\bm{p}^{T},(\bm{p}^{'})^{T})^{T}$ denotes the complete optimization variables in our algorithm. Specifically, we focus on the expansion for the denominator of the objective function in $\bm{P2}$ and set an initial power vector $\bar{\bm{p}}_{0}$ to calculate a group of coefficients. Then, the denominator is transformed into a monomial $\tilde{g}(\bar{\bm{p}})$ with calculated coefficients. This problem is GP in standard form, and we can modify it to a GP in convex form and use the interior point method to solve. With the optimal solution $\bar{\bm{p}^{*}}$, we solve different GP constantly to improve the accuracy for the approximation of original problem. Notice that the Karush-Kuhn-Tucker (KKT) conditions is satisfied obviously as we take derivatives of $g(\bar{\bm{p}})$ and $\tilde{g}(\bar{\bm{p}})$, this is verified as the best local monomial approximation \cite{30}. The algorithm is convergent to the optimal solution of $\bm{P2}$, and we will show numerical results in simulation section. 
The proposed algorithm is summarized in Algorithm~\ref{algorithm1}, and we define
\begin{equation}
\begin{aligned}
&\prod_{k=1}^{K}\prod_{\ell=1}^{J_{D}}g_{k}(\bar{\bm{p}})g_{\ell}^{'}(\bar{\bm{p}})\triangleq\sum_{s}u_{s}(\bar{\bm{p}})\\
&\triangleq\sum_{s}c_{s}(P_{1}^{a_{s}^{1}}\cdots P_{J}^{a_{s}^{J}})[(P^{'}_{1})^{b_{s}^{1}}\cdots (P^{'}_{J_{D}})^{b_{s}^{J_{D}}}],
\end{aligned}
\label{eq47}
\end{equation}
where $c_s>0$, and $a_{s}^{j},b_{s}^{\ell}\in \bm{R}$.

\begin{algorithm}[h]
	\caption{Iterative GP power allocation algorithm.}
	\begin{algorithmic}[1]
		\STATE Initialize $\bar{\bm{p}}=\bar{\bm{p}}_{0}$ and set the maximum iterative number $T^{\max}$.
		\STATE Set $I=1$.     
		\WHILE {$I\leq T^{\max}$}
		\STATE Compute coefficients $\beta_{s}$ by\\ $\beta_{s}=u_{s}(\bar{\bm{p}})/[\prod_{k=1}^{K}\prod_{\ell=1}^{J_{D}}g_{k}(\bar{\bm{p}})g_{\ell}^{'}(\bar{\bm{p}})]$.
		\STATE Substitute the denominator of the objective function in $\bm{P2}$ by $\tilde{g}(\bar{\bm{p}})$, where\\ $\tilde{g}(\bm{p},\bm{p^{'}})=\prod_{s}\left[\frac{u_{s}(\bm{p},\bm{p^{'}})}{\beta_{s}}\right]^{\beta_{s}}$.
		\STATE Solve the new GP optimization problem in convex optimization methods with the objective function $\frac{\prod_{k=1}^{K}\prod_{\ell=1}^{J_{D}}f_{k}(\bm{p^{'}})f_{\ell}^{'}(\bm{p})}{\tilde{g}(\bm{p},\bm{p^{'}})}$ and the same constraints as in $\bm{P2}$.
		\STATE Update $\bar{\bm{p}}={\bar{\bm{p}}^{*}}$, where ${\bar{\bm{p}}^{*}}$ is the optimal solution of the new GP problem.
		\STATE $I=I+1$.
		\ENDWHILE
		\RETURN The optimal solution ${\bar{\bm{p}}^{*}}$ of $\bm{P2}$.
	\end{algorithmic}
	\label{algorithm1}
\end{algorithm}

The main complexity in Algorithm~\ref{algorithm1} for each iteration reflects on the computation of coefficients $\bm{\beta}$ and the process of solving GP. Actually, the computation of new function $\tilde{g}(\bm{p},\bm{p^{'}})$ requires $(K+J_{D})(d_{c}+1)+\sum_{s}(J+J_{D}+3)$ multiplications and $\sum_{s}(J+J_{D}+1)$ exponentiations. Besides, GP is usually solved by interior point methods, which is proved to have polynomial time complexity \cite{29,37}.
\section{Simulation Results}
\label{Simulation}
In this section, we perform simulations to demonstrate the performance of our power allocation algorithm. Notice that uplink SCMA system with large users can be resolved in small-scale SCMA structure with same overloading feature, the SCMA scheme with $J=6$, $K=4$, and $N=2$ is claimed to be the basic regular SCMA system. Our simulation scenario utilizes the basic SCMA scheme with the same factor graph in figure \ref{fig2}. Considering the decoding performance of D2D communications in proposed network \cite{23}, we set no more than two D2D pairs in this scenario. Besides, notice that the small modulation and simple channel coding are preferable as the devices in massive machine communications are low-complexity \cite{35}, we adopt Quadrature Phase Shift Keying (QPSK) for cellular users here. Thus, each cellular user's total power should be equally allocated to two subcarriers according to the SCMA codebook structure. The channels for both cellular communications and D2D communications are generated with a normalized Rayleigh fading component and a distance-dependent path loss model. Table~\ref{table1} summarizes the simulation parameters \cite{15,36}. Different scenes and methods are compared to prove that this algorithm is appropriate for this specific optimization problem. 

\begin{table}
	\centering
	\caption{Simulation parameters.}
	\begin{tabular}{|c|c|}
		\hline 
		\scriptsize No. CU, $J$ & \scriptsize 6\\ 
		\hline 
		\scriptsize No. subcarriers, $K$ & \scriptsize 4\\ 
		\hline 
		\scriptsize No. non-zero dimension, $N$ & \scriptsize 2\\ 
		\hline 
		\scriptsize Noise power, $N_{0}$ & \scriptsize -174dBm/Hz\\ 
		\hline 
		\scriptsize Power limitation of CU, $P_{0}$ & \scriptsize 30dBm\\ 
		\hline 
		\scriptsize Power limitation of DU, $P_{0}^{'}$ & \scriptsize 30dBm\\ 
		\hline 
		\scriptsize SINR of CU, $\gamma_{0}^{c}$ & \scriptsize 0dB\\ 
		\hline 
		\scriptsize SINR of DU, $\gamma_{0}^{d}$ & \scriptsize 10dB\\ 
		\hline
		\scriptsize Bandwidth, $B$ & \scriptsize 180kHz\\
		\hline
		\scriptsize Initial Power of CU, $p_{0}$ & \scriptsize $P_{0}/2$\\
		\hline  
		\scriptsize Initial Power of DU, $p_{0}^{'}$ & \scriptsize $P_{0}^{'}/2$\\
		\hline
		\scriptsize Cell radius & \scriptsize 500m\\ 
		\hline
		\scriptsize Distance of each D2D pair & \scriptsize 1$\sim$20m\\ 
		\hline 
		\scriptsize Path loss Model of CU & \scriptsize 37.6log10($d_C$[km])+128.1\\ 
		\hline 
		\scriptsize Path loss Model of D2D & \scriptsize 40log10($d_D$[km])+148\\
		\hline    
	\end{tabular} 
	\label{table1}
\end{table}

\subsection{Convergence}
\label{convergence}
The convergence of this algorithm is shown from figure \ref{fig3} to figure \ref{fig5}. Figure \ref{fig3} illustrates the power allocation versus number of iterations by Algorithm~1, where CU and DU denote cellular user and D2D user, respectively. The number of D2D pairs $J_{D}$ is $1$ in figure \ref{fig3} with a specific numerical result from the tests. The figure shows that the power allocation algorithm converges very fast and achieves the optimal power only after three iterations. We can analyze that after few series of approximation, the solutions converge to a point satisfying the KKT conditions of the origin problem in $\bm{P2}$. Thus, transforming complementary GP into standard GP by this algorithm often performs a fast convergence rate. Besides, two users are allocated with maximal power from all seven users as they may suffer from relatively small interference. 

Moreover, we consider the scenarios with two pairs of D2D users, which of course will improve the computational complexity of the optimization problem as the variables increase. We can observe from figure \ref{fig4} that only one CU is allocated with maximal power as the interference among cellular users and D2D users increases. Besides, the power allocation converges very close to the optimal power allocation by the $2$th iterations, which still proves that this algorithm performs a fast convergence rate.

\begin{figure}[h]
	\centering  
	\includegraphics[width=0.8\linewidth]{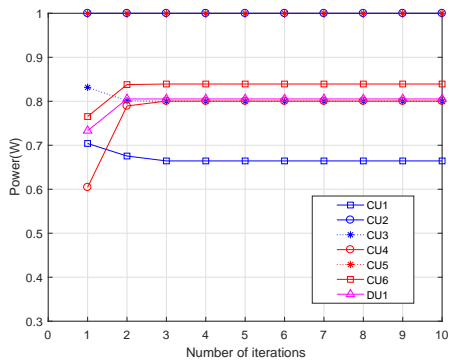}  
	\caption{Convergence of power allocation, $J_{D}=1$.} 
	\label{fig3}
\end{figure}

\begin{figure}[h]
	\centering  
	\includegraphics[width=0.8\linewidth]{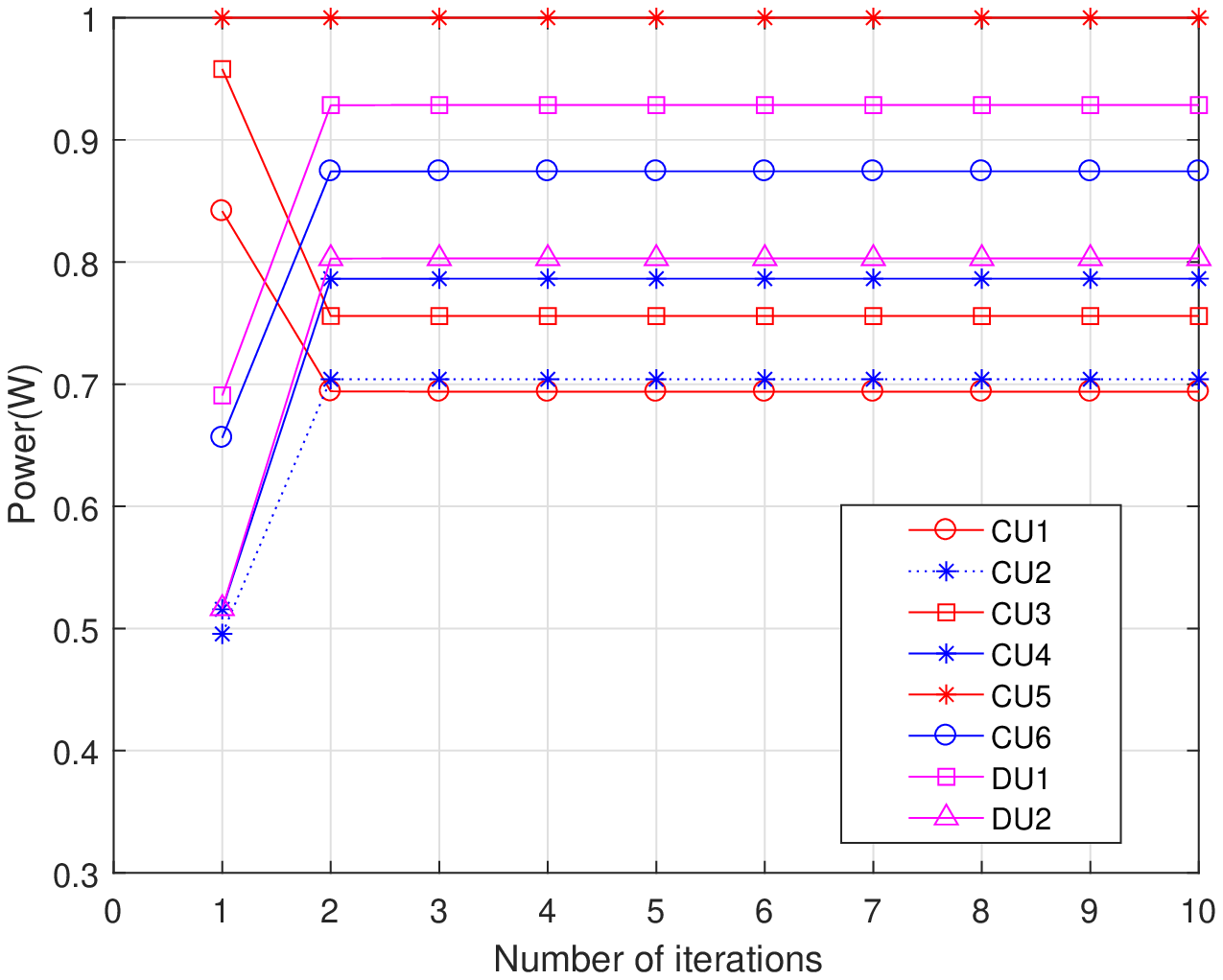}  
	\caption{Convergence of power allocation, $J_{D}=2$.} 
	\label{fig4}
\end{figure}

\begin{figure}[h]
	\centering  
	\includegraphics[width=0.8\linewidth]{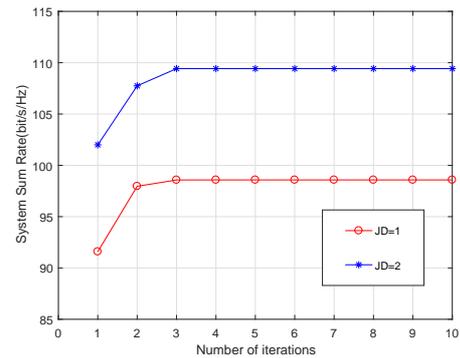}  
	\caption{Convergence of system sum rate.} 
	\label{fig5}
\end{figure}
The sum rate versus number of iterations is shown in figure \ref{fig5}, which compares the scenarios with different D2D pairs. The algorithm always converges within five iterations while the rate achieves maximum after three iterations from figure \ref{fig5} in both scenarios. From this aspect, the algorithm is also proved to have a better convergence rate. Besides, the increase of D2D pairs makes the performance of sum rate rises about $1.2\%$ while other parameters remain same.
\begin{figure}[h]
	\centering  
	\includegraphics[width=0.8\linewidth]{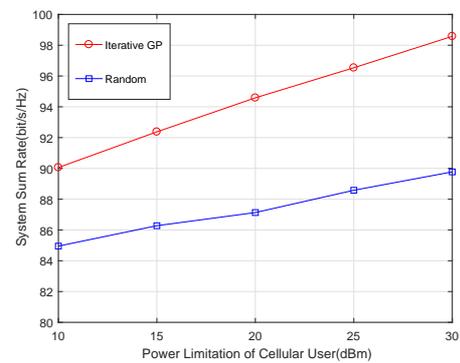}  
	\caption{Algorithm performance comparison versus cellular power limitation with $J_{D}=1$.} 
	\label{fig6}
\end{figure}

\begin{figure}[h]
	\centering  
	\includegraphics[width=0.8\linewidth]{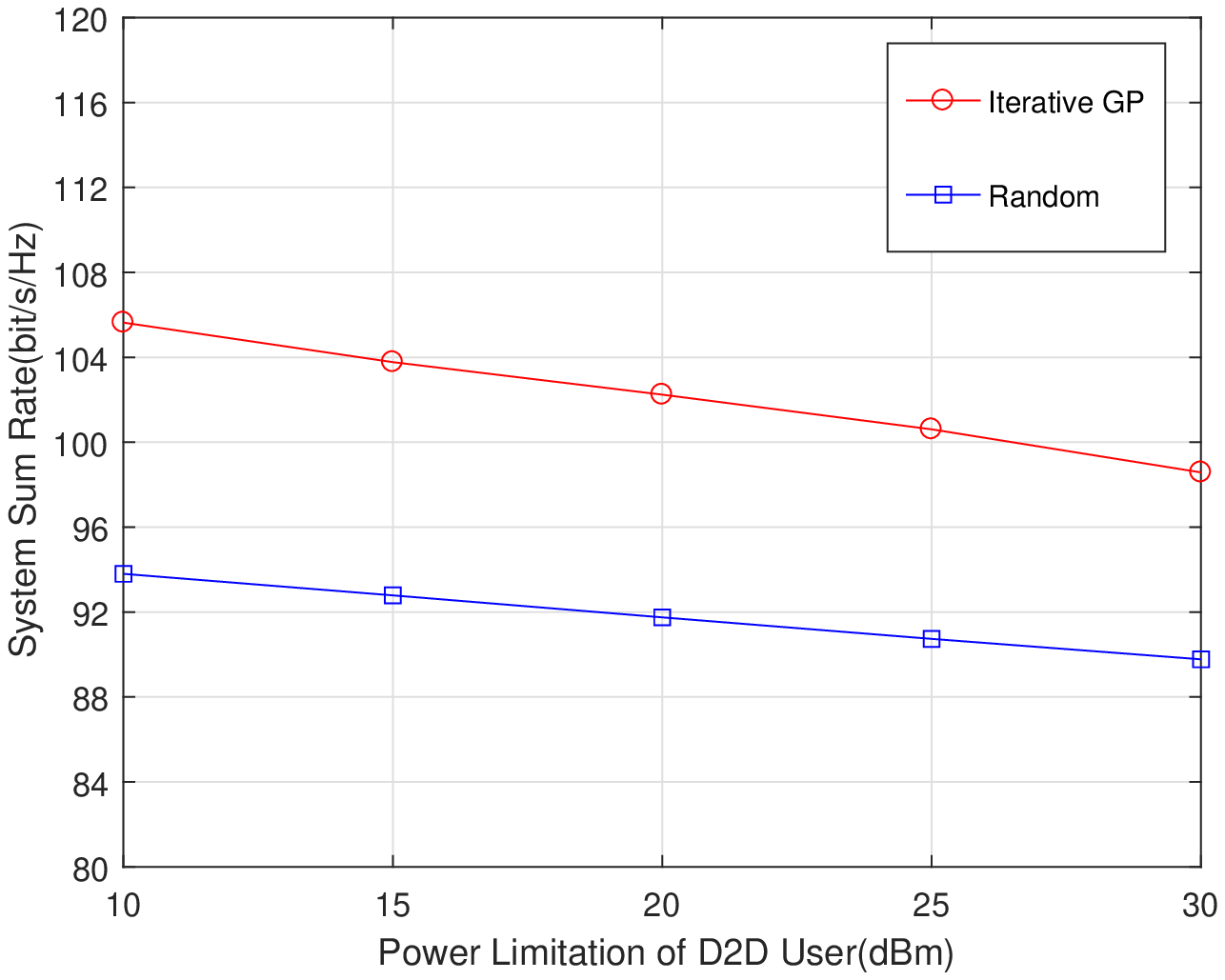}  
	\caption{Algorithm performance comparison versus D2D's power limitation with $J_{D}=1$.} 
	\label{fig7}
\end{figure}

\begin{figure}[h]
	\centering  
	\includegraphics[width=0.8\linewidth]{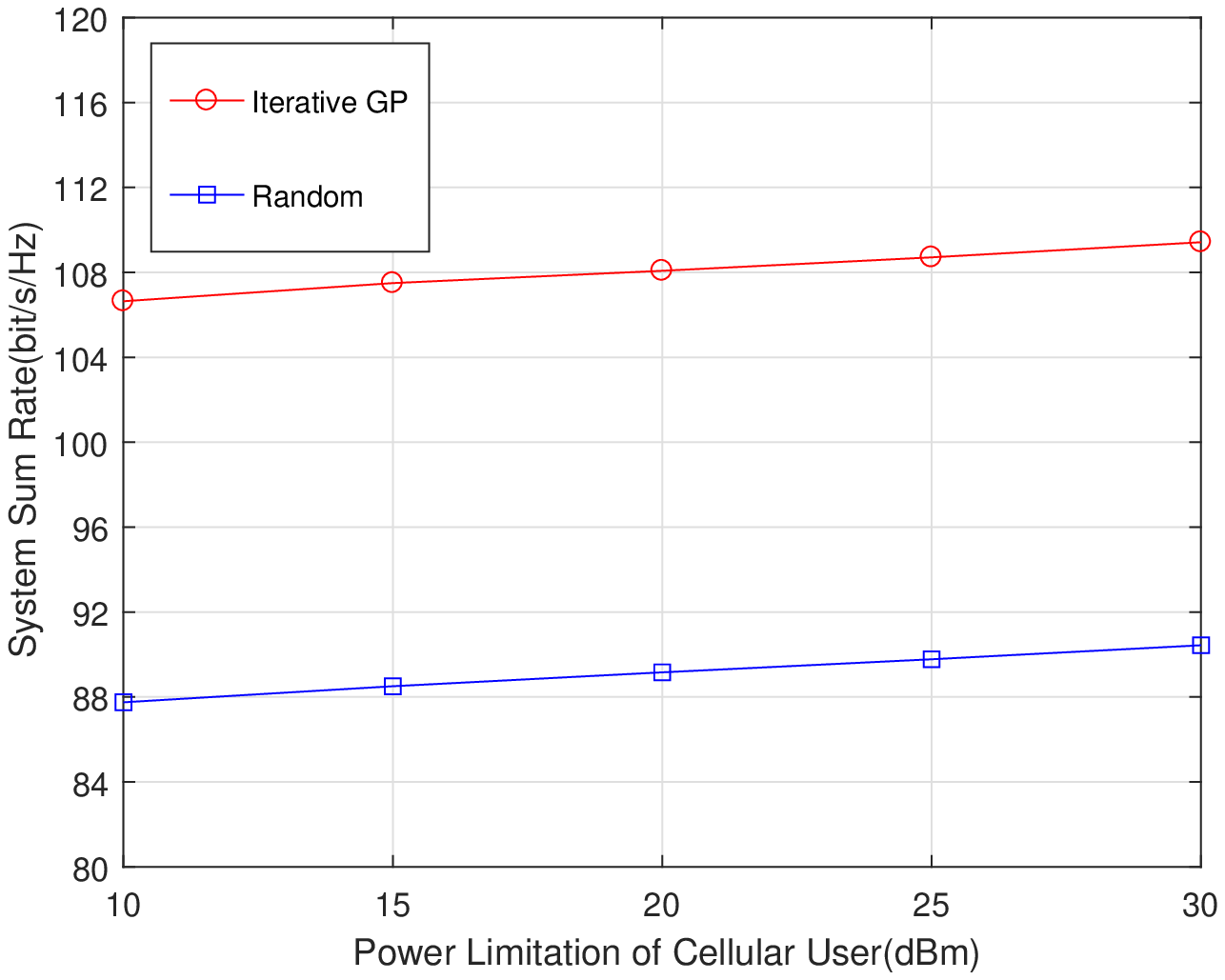}  
	\caption{Algorithm performance comparison versus cellular power limitation with $J_{D}=2$.} 
	\label{fig8}
\end{figure}

\begin{figure}[h]
	\centering  
	\includegraphics[width=0.8\linewidth]{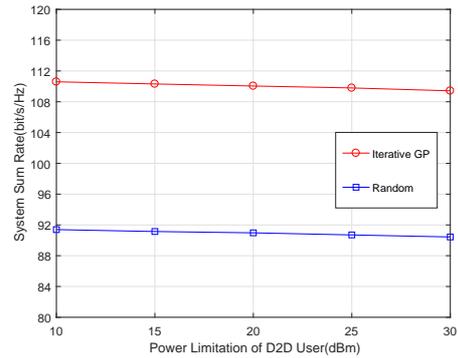}  
	\caption{Algorithm performance comparison versus D2D's power limitation with $J_{D}=2$.} 
	\label{fig9}
\end{figure}

\subsection{Algorithm Performance}
\label{performance}
We compare the performance of the proposed algorithm with random power allocation from figure \ref{fig6} to figure \ref{fig9}. We set different power limitation values for both cellular users and D2D users to observe the performance of final system sum rate. The sum rate versus the power limitation of cellular users $P_{0}$ is shown in figure \ref{fig6}, while figure \ref{fig7} illustrates system sum rate versus power limitation of D2D user, $P_{0}^{'}$. Both the results in figure \ref{fig6} and figure \ref{fig7} utilize the scenario of $J_D=1$. The two figures prove that the proposed iterative GP based algorithm achieves better sum rate performance than the random power allocation, which demonstrates the advantage of the proposed algorithm. That is because the iterative algorithm always converges to the optimal solution to the original optimization problem. Besides, it can be founded that the increase of threshold for cellular users will improve the performance of sum rate. However, when the threshold for D2D users increases, the sum rate will decrease. The reason is that the rate of cellular users always has larger effect on the system sum rate, and the less D2D's power allocation could give improvement to the whole system while the communication of D2D users must be guaranteed. 

Figure \ref{fig8} and figure \ref{fig9} show the similar comparison for the system sum rate performance with the scenario of $J_D=2$. We can observe that the gap of performances between proposed algorithm and random scheme becomes larger compared to the scenario of $J_D=1$. It means that the effectiveness of our proposed algorithm is more prominent with much mutual interference case. Besides, the variety of either cellular or D2D's threshold doesn't provide obvious change to the performance of system sum rate. The probable reason is that cellular communications and D2D communications have made equal effect on the performance of system sum rate with much mutual interference.
\section{Conclusion}
\label{Conclusion}
We investigate the hybrid cellular network with D2D communications, where SCMA schemes is employed for cellular users. To improve the performance of system sum rate region, we first derive the capacity for cellular users as well as D2D users. Then we consider the joint power optimization problem to maximize the system sum rate and propose a GP based iterative algorithm. Simulation results prove the fast convergence and show the much improvement of proposed algorithm compared with random power allocation. With the increasing number of D2D users, the proposed algorithm performs prominent for the hybrid network model.
\section*{ACKNOWLEDGEMENT}
\label{Ackonwledgement}
This paper is supported by National key project 2018YFB1801102 and 2020YFB1807700, by NSFC 62071296, and STCSM 20JC1416502, 22JC1404000.
\section*{Appendix}
\label{appendix}
Here is the proof of Lemma~\ref{lemma3}. The arithmetic-geometric mean (AM-GM) inequality is shown as
\begin{equation}
\sum_{\ell}\beta_{\ell}x_{\ell}\geq\prod_{\ell}x_{\ell}^{\beta_{\ell}},
\label{eq48}
\end{equation}
where $x_{\ell}>0$, $\beta_{\ell}\geq0$, and $\sum_{\ell}\beta_{\ell}=1$, for $\forall \ell$. We utilize $y_{\ell}=\beta_{\ell}x_{\ell}$ and rewrite Eq.~\eqref{eq48} as
\begin{equation}
\sum_{\ell}y_{\ell}\geq\prod_{\ell}(\frac{y_{\ell}}{\beta_{\ell}})^{\beta_{\ell}},
\label{eq49}
\end{equation}  
with equality for $\beta_{\ell}=y_{\ell}/\sum_{\ell}y_{\ell}$, for $\forall \ell$. This equality can be extended to the case of general monomial and polynomial \cite{30}.

\bibliographystyle{gbt7714-numerical}
\bibliography{myref}
\newpage
\biographies
\begin{CCJNLbiography}{lyk.eps}{Yukai Liu}
received the B.S. degree from Shanghai Jiao Tong University (SJTU) in 2016. He is currently pursuing the Ph.D. degree on Information and Communication Engineering from SJTU. His current research interests include wireless communications, with particular focus on multiple access and device-to-device communications.
\end{CCJNLbiography}

\begin{CCJNLbiography}{cw.eps}{Wen Chen}
 is a tenured professor at Shanghai Jiao Tong University. He is a fellow of Chinese Institute of Electronics, and a distinguished lecturer of IEEE Communications Society and IEEE Vehicular Technology Society. He is an Editor of IEEE Transactions on Wireless Communications, IEEE Transactions on Communications, IEEE Access and IEEE Open Journal of Vehicular Technology, and the Shanghai chapter chair of IEEE Vehicular Technology Society. His research interests include reconfigurable meta-surface, multiple access, wireless AI, and green networks. He has published 100+ articles in IEEE journals and more than 120 papers in IEEE Conferences, with citations more than 8000 in Google scholar.
\end{CCJNLbiography}

\end{document}